\documentclass[runningheads]{llncs}
\usepackage[utf8]{inputenc}
\usepackage{graphicx}

\usepackage{amsmath}

\usepackage{amsthm}

\usepackage{amssymb}
\usepackage{xcolor}
\usepackage{hyperref}
\hypersetup{
	colorlinks,
	linkcolor={green!70!black},
	citecolor={blue!70!black},
	urlcolor={blue!80!black}
}

\usepackage{makecell}

\renewcommand{\O}{\mathcal{O}}

\newtheorem{observation}{Observation}

\begin{document}

\title{Optimal Window Queries on Line Segments using the Trapezoidal Search DAG\thanks{See \url{https://arxiv.org/abs/2111.07024} for the full version of this paper.}}
\titlerunning{Optimal Window Queries on Line Segments using the TSD}

\author{Milutin Brankovic\inst{1} \and Martin~P. Seybold\inst{2}}
\institute{
University of Sydney, School of Computer Science, Sydney, Australia\\
\email{mbra7655@uni.sydney.edu.au}
\and
University of Vienna, Faculty of Computer Science, Vienna, Austria
\email{martin.seybold@univie.ac.at}
}


\maketitle
\noindent
\makebox[\linewidth]{\small Published: \href{https://doi.org/10.1007/978-3-031-22105-7_46}{01 January 2023}}

\begin{abstract} 
We propose a new query application for the well-known Trapezoidal Search DAG (TSD) of a set of $n$~line segments in the plane, where queries are allowed to be {\em vertical line segments}.

We show that a simple Depth-First Search reports the $k$ trapezoids that are intersected by the query segment in ${\cal O}(k+\log n)$ expected time, regardless of the spatial location of the query.
This bound is optimal and matches known data structures with $\O(n)$ size. 
In the important case of edges from a connected, planar graph, our simplistic approach yields an expected ${\cal O}(n \log^*\!n)$ construction time, which improves on the construction time of known structures for vertical segment-queries. 
Also for connected input, a simple extension allows the TSD approach to directly answer axis-aligned window-queries in ${\cal O}(k + \log n)$ expected time, where $k$ is the result size.

\end{abstract}


\section{Introduction}
Intersection reporting problems can be stated as follows:
Given a set of geometric objects and a query object, report all elements of the set that intersect the query object.
A classic problem is to organize line segments in the plane such that segments that intersect an axis-aligned query window can be reported quickly.

It is well-known that the \emph{window-query} problem can be reduced to answering the two following problems:
Find all endpoints that are inside the query window (i.e. axis-orthogonal range queries) and find all segments that cross the window boundary (i.e. vertical and horizontal boundary queries). 
Clearly these two results can be assembled to the final set in time that is proportional to the output of the reporting query.
Using Chazelle's Range Trees~\cite{Chazelle86}
, one can determine the endpoints in $\O(k+\log n)$ time.
In practice however, simpler data structures with \emph{linear size} are dominant in applications like the KD-Tree~\cite{Bentley75}, the R-Tree~\cite{Guttman84}, and Delaunay Triangulations.

For intersections with the window boundary it suffices to solve the \emph{vertical queries}, since the horizontal boundary queries can be solved analogously (see~\cite[Chapter~$10$]{dutchBook}).
For sets of horizontal segments, one can built in $\O(n \log n)$ time an Interval Tree~\cite{edelsbrunner1980dynamic} based structure, using McCreight's Priority Search Trees~\cite{McCreight85}, of size $O(n)$ that supports vertical segment-queries in $\O(k + \log^2 n)$ time (see Sections $10.1$-$2$ 
in~\cite{dutchBook}).
For non-crossing line segments, Bentley's~\cite{BentleySegTree} well known  Segment Data Structure~(SDS) answers vertical segment-queries in $\O(k + \log^2 n)$ time, however the tree structures have $\O(n \log n)$ size~\cite[Chapter~$10.3$]{dutchBook}.

Our method is based on the randomized incremental construction of a planar subdivision. 
\emph{Trapezoidations} are generally defined for $x$-monotonous curves~\cite{Mulmuley91,Seidel93-cccg} (e.g. line segments) where each pair of curves is only allowed to share at most a finite number of points (i.e. endpoints or intersection points where curves cross) and no three curves may intersect in the same point.
Given such input, the planar subdivision is obtained by shooting a vertical ray up and down from every end and intersection point until it meets the first curve or the domain boundary, which results in $\O(n + I)$ (pseudo) trapezoidal faces where $I$ is the number of intersection points.
In particular, overlapping curves are excluded.
Though one can generally avoid degeneracies by slightly perturbing the input, it is \emph{necessary} for exact query reporting to explicitly handle degenerate input.
There are three well-known randomized incremental constructions to compute trapezoidations, and thereby report all $I$ intersection points, in $\O(I+n\log n)$ expected time.
The methods mainly differ in their strategy how the trapezoidal regions, that are affected by the next segment insertion, are located.
These are Clarkson and Shor's Conflict Graph method~\cite{ClarksonS89}, Mulmuley's Conflict Point+Walk method~\cite{Mulmuley90}, and Seidel's Locate+Walk method~\cite{Seidel91} that introduced the TSD.

\paragraph*{Related Work}
The walk `along the segment' in a trapezoidation may well require to check a large number of trapezoidal faces that are not intersected by the segment whenever the walk {\em crosses} a segment.
Utilizing randomization, Mulmuley showed that the total expected time of `face-transitions' from walks, over the entire construction, is $\O(I + n \log n)$.
(See Figure~$1$(f) for the illustration of one face-transition and Theorem~$5$ in \cite{Mulmuley90}.)
However, face-transition cost is not negligible if the walk's segment is a user specified query. 
In the point-query setting, the TSD is well-known to return the face of the trapezoidation that contains the query point, with high probability, in $\O(\log n)$ time~\cite[Section~$6.2$-$4$]{dutchBook}.
Thus, vertical ray-shooting queries, i.e. reporting the segment immediately above or below a query point, are supported.
Seidel~\cite[Thm.~2-3]{Seidel91} also showed that the expected TSD construction time can be made $\O(n \log^*\!n + \gamma\log n)$, where $\log^*$ dentotes the \emph{iterated logarithm}\footnote{I.e. the number of times the $\log$-function must be applied to obtain a result $\leq 1$.}, if the input is a plane graph with $\gamma$ connected components (e.g. $\gamma=1$ for simple polygons).

For $n$ non-crossing segments, Chazelle showed in~\cite[Sec.~$4.2$]{Chazelle86} that preprocessing with a `radar-beam' sweep-line allows to reduce the query problem to a related instance on $n$ non-crossing, \emph{horizontal} segments, using $\O(n)$ space and $\O(n \log n)$ time.
Moreover, Section~4.1 of his work shows that the trapezoidation of horizontal segments can be augmented with $\O(n)$ edges to a planar subdivision, called `Hive-Graph', such that a vertical walk has $\O(1)$ cost for each face-transition.
Finally, a point-location structure for the hive-graph allows answering vertical-queries in $\O(k + \log n)$ time.
For horizontal segments, the optimal query bound can also be achieved based on Persistence~\cite{DriscollST-persistence} or based on Random Sampling for constructing a two-layer data structure that stores in each leaf a `Conflict List' (see Section~4.3 in~\cite{Chan13}).

\paragraph{Contribution and Paper Outline}
In Sections~\ref{sec:TSD-construction}-\ref{sec:proof}, we show that the $\O(k + \log n)$ bound for vertical \emph{segment-queries} already holds for a Depth-First Search~(DFS) in the TSD of the input segments.
Thus, all preprocessing phases and additional, linear sized, data structures of known techniques are not required.
See Table~\ref{tab:vertical-query-structures} for a comparison of our TSD based approach with known data structures for vertical segment-queries.

\begin{table}[t]
    \centering
\begin{tabular}{l|l|l|l|c|c}
Data Structure         & Size           & Construction & Query Time      & Techn. &  Input \\ \hline
SDS~\cite{BentleySegTree}                    & $\O(n\log n)$  & $\O(n \log n)$  & $\O(k+\log^2 n)$ & Det. & non-crossing\\
Intervall Tree~\cite{McCreight85}   & $\O(n)$        & $\O(n \log n)$  & $\O(k+\log^2 n)$ & Det. & horizontal\\
Hive-Graph~\cite{Chazelle86}           & $\O(n)$  & $\O(n \log n)$  & $\O(k+\log n)$ & Det. & non-crossing\\
Persistentce~\cite{DriscollST-persistence,SarnakT86}  & $\O(n)$        & $\O(n \log n)$  & $\O(k+\log n)$ & Det. & non-crossing\\
Conflict Lists~\cite{Chan13}            & $\O(n)$  & $\O(n \log \frac{n}{\log n})$  & $\O(k+\log n)$ & Rnd. & horizontal \\

\hline 
~ & $\O(n + I)$ & $\O(n \log n + I)$ & ~ &  ~ & general \\
TSD               & $\O(n)$ & $\O(n \log n)$ & $\O(k+\log n)$ & Rnd. &   non-crossing  \\
~               & $\O(n)$ & $\O(n \log^*\!n)$ & ~ & ~ &  connected, planar \\
\hline
\end{tabular}
    \caption{Data structures for reporting line segments that cross a {\em vertical query segment}. $k$ denotes the output size, $I$ the number of intersection points ($I=0$ for non-crossing segments), and $\log^*$ the iterated logarithm.
    All TSD expectation bounds are over the randomness of the data structure.}
\label{tab:vertical-query-structures}
\end{table}

The simplicity of our method yields an improved construction time bound of $\O(n \log^*\!n)$ for the important case where segments are edges from a connected, planar graph. 
Though not requiring connectivity, the best known construction was $\O\!\left(n \log\frac{n}{\log n}\right)$ expected time, due to the random sampling of conflict lists (cf. \cite[Sec.~4.3]{Chan13}).
In Section~\ref{sec:connected-planar}, we give a simple extension that allows answering axis-aligned \emph{window-queries} on connected, planar graphs, directly with our TSD approach.

\section{Proposed TSD for General Input and the DFS Query}
\label{sec:TSD-construction}
To support exact reporting on a general input set $S$ with trapezoidation and TSD, we need to handle all input degeneracies explicitly.
In basic trapezoidations, non-crossing segments are only allowed to share a common end-point and crossing segments are assumed to have distinct intersection points (cf. Section~$6.3$ in~\cite{dutchBook}).
However, general input may well have overlapping segments and multiple intersections in the same point.
Known, exact implementations~\cite{BerberichFHMW10,HemmerKH12} treat overlapping segments at cost of information loss.
I.e. sweep line pre-processing removes degeneracies by merging and splitting overlapping segments.
It is then not possible to correctly answer queries on the resulting data structure without additional work. 
For example, the ply of a single result fragment can already be $\Omega(n)$ and there are $\Omega(n)$ such fragments in general, leading to a \emph{quadratic} space cost.
To our knowledge, there are no other publicly available implementations of TSDs that handle overlapping segments in the following way:

We define a \textit{segment bundle} to be the set of segments that are contained in the same support line and apply an implicit, infinitesimal vertical translation to distinguish between overlapping segments.
This is achieved by imposing a total ordering on the set of segments.
When two segments overlap, the `larger' segment in the ordering is considered to be higher in $y$-direction.
Bundles in the TSD give rise to trapezoids of zero geometric area.
%
Naïvely, it may appear that intersections of bundles need no special treatment and that one may treat the intersection of sets of overlapping segments no differently than regular segments. 
In such an approach however, there seems to be no clear tie-breaking between the resulting vertical partitions, since all are defined by the same geometric point.
Moreover, this approach gives rise to several trapezoids that only consist of a single point. 
Instead, our approach uses a single vertical partition for an intersection of bundles. 
We consider two segments, e.g. the `oldest' from a bundle, to be the defining segments of the intersection point of the vertical partition.
Other segments are considered to be `threaded' through the present intersection point.
Conceptually, this has the effect of `bending' the segments, resulting in degenerate trapezoids with zero height.
Intersecting bundles give rise to spatially empty $2$-gons in the trapezoidation.
We use TSD nodes that have arity $\leq4$ and store an explicit representation of their trapezoid boundary (i.e. pointers to input segments and intersection points).
See Figure~\ref{fig:seg_insert_multiple_faces} for an example.
(See also Section~$3$ and Theorem~$3$ in~\cite{Seidel91} for the analysis of the construction time.)


Using this tie-breaking, a vertical query segment $q$ that passes through a multiple intersection point $p$ intersects only the interior of two of the trapezoids that are incident to $p$.
We thus store with each intersection point of the trapezoidation a list of all segments that are threaded through it. 
For example, the intersection point in Figure~\ref{fig:seg_insert_multiple_faces} (square nodes) points to a list with the elements $\{s_1,s_2,s_3\}$.
This enables exact reporting queries \emph{solely} based on the TSD.

\begin{figure}[t]
    \centering
    \includegraphics[width=\columnwidth]{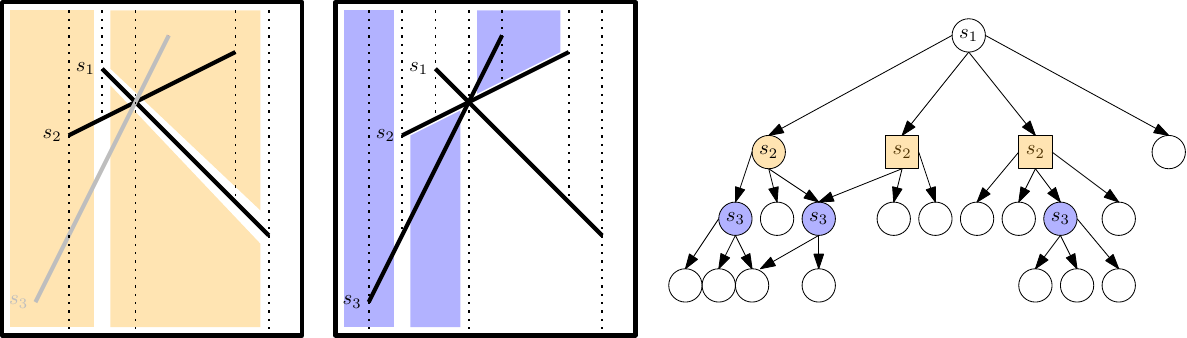}
    \caption{Incremental construction of the TSD (right) for the insertion order $(s_1,s_2,s_3)$ and the respective trapezoidations (left).
    The TSD has $14$ leaves, one for each face in the final trapezoidation.
    }
    \label{fig:seg_insert_multiple_faces}
\end{figure}

 \subsubsection{Don't Walk: Vertical Segment-Queries via DFS}
Starting at the root, our Depth-First Search visits at each node $\leq 2$ children, those whose trapezoidal region intersect the vertical segment-query $q$. 
The algorithm reports all reached TSD leaves, i.e. faces from the trapezoidation of $S$. 
The number of faces is $\leq 1+k$, where $k$ is the number of segments from $S$ that intersect $q$.
All intersected segments are obtained by inspecting the boundary of the result trapezoids.

\section{Expected Time of Vertical Segment-Queries} \label{sec:proof}

In this section, we show the following bound on the time complexity of our segment-query.

\begin{theorem}\label{query-bound}
Let $S$ be a set of $n$ segments, $q$ a vertical segment, and $k=|\{s \in S: s\cap q \neq \emptyset\}|$.
The expected number of TSD nodes, whose regions intersect $q$, is $\O(k + \log n)$.
\end{theorem}

All expectations are over the permutations of $S$.
To prove this, we require some definitions.
Let $q_0$ and $q_1$ denote the bottom and top endpoint of $q$, respectively.
Let $\mathcal{T}$ denote the set of trapezoids in the trapezoidation of $S$. Note that these trapezoids are in one-to-one relation with leaves of any TSD over $S$.
For $q$, we denote the result set by $R(q) = \{\Delta \in \mathcal{T} : \Delta \cap q \neq \emptyset\}$ and define the random variable $V_q$ to be the set of nodes of the TSD visited by the segment-query. 
For $\Delta \in R(q)$, let $q(\Delta)$ be the point in $q\cap \Delta$ with the lowest $y$-coordinate.
That is $q(\Delta)$ is either on the bottom segment of the trapezoid or $q_0$.
Let $K_q$ be the multi-set $\{q(\Delta):\Delta\in R(q)\}$ of those points.
($K_q$ is an ordinary set if all $\Delta \in R(q)$ have non-zero area.)
First observe that node set $V_q$ is the union of the search paths of the points in $K_q$.

\begin{observation}\label{tree-characterization}
    $V_q$ is exactly the set of TSD nodes whose region contains a point in $K_q$.
\end{observation}

Bounding the expected running time of the DFS query algorithm amounts to bounding the expected size of $V_q$ over the permutations of $S$.
A coarse bound is easily given by summing over the expected sizes of the search paths of each of the $\leq k+1$ points in $K_q$.
It is well-known that the length of any search path is $\mathcal{O}(\log n)$ in expectation, moreover this bound holds with high probability~\cite[Section~$6.4$]{dutchBook}.
This gives an expected query time bound of $\mathcal{O}(k\cdot \log n)$. 
We improve on this bound with a finer analysis that reduces overcounting.
For example, the DFS search visits the TSD root only once.

Our approach is based on the backwards analysis technique used to bound the expected running time of a point-query.
We briefly summarize the arguments given in Chapter~$6$ of \cite{dutchBook} below, before extending them to our problem.
Given a point-query $p$, we consider the sequence of leaves $A_0,A_1,\ldots ,A_n$ whose region contains $p$ at the respective step of the randomized incremental construction of the TSD.
That is, after $i$ segments had been inserted into the structure, $A_i$ was the leaf whose trapezoidal region contained $p$.
Note that each $A_i$ remains in the data structure, possibly as an internal node, even after more segments are inserted.
Observe that the number of node visits in a point location query for $p$ is equal\footnote{Our TSD construction in Section~\ref{sec:TSD-construction} uses one node of arity four per trapezoid in $\cal T$.} to the number of indices $i>0$ such that $A_i \neq A_{i-1}$.
For $i>0$, we define the random variables
\begin{align*}
X_i = \begin{cases}
   1,& \text{if } A_i \neq A_{i-1}\\
    0,              & \text{otherwise}
\end{cases} \quad.
\end{align*}

It is easy to see that the expected number of visited nodes for $p$ is exactly $\sum_{i=1}^{n}\mathbb{E}[X_i]$.
Observe that at most four segments bound the region of each $A_i$.
Thus $A_i \neq A_{i-1}$ if and only if one of the bounding segments of $A_i$ is the $i^{\text{th}}$ segment of the random order.
Since each of the first $i$ segments are equally likely to be at the $i^\text{th}$ position, we have that the probability that $A_i \neq A_{i-1}$, and therefore the expected value of $X_i$, is $\leq 4/i$.
It follows that the expected time of the point-location query for $p$ is 
        $\sum_{i=1}^{n} \mathcal{O}(\mathbb{E}[X_i]) = \mathcal{O}\left(\sum_{i=1}^{n} 4/i\right)
        = \mathcal{O}(\ln n) $.

We now extend this analysis technique for our segment-query bound.
Let $\pi:S\rightarrow \{1,...,n\}$ denote a permutation which determines an insertion order of segments into the TSD. 
We define the priority of a segment $s\in S$ to be $\pi(s)$. 
We define the priority of a trapezoid $\Delta \in \mathcal{T}$, denoted $\pi(\Delta)$, to be the priority of the segment which defines its bottom boundary. 
The priority of an internal TSD node is the priority of the segment that destroys it. 
One may use the backwards analysis argument above to bound the expected number of internal nodes $V' \subseteq V_q$ which contain a point in $K_q$. 
Clearly, the number of leaf nodes in $V_q$ is $|K_q|\leq k+1$ already.
Our proof bounds the expected number of visited nodes in separate priority value intervals.
Specifically, we partition the set of priority values $\{1, \ldots ,n\}$ into intervals of exponentially increasing length.
For each of the intervals, we bound the expected number of nodes from $V_q$, whose priority falls into the interval.
Counting this way turns out to reduce the overestimation sufficiently to obtain the desired ${\cal O}(k+\log n)$ bound on the expected query time.

For integer $\alpha \in \{1,\ldots,n\}$, we call the set of trapezoids $\mathcal{T}_{\leq \alpha} = \{\Delta \in \mathcal{T} : \pi(\Delta) \leq \alpha\}$ from trapezoidation $\cal T$ the $\alpha$-sample.
Furthermore, we say the point $q(\Delta)$ is $\alpha$-active, if its trapezoid $\Delta$ is in the $\alpha$-sample or $q(\Delta)=q_0$.

\begin{lemma} 
Let $S$ be a set of $n$ segments, $q$ a vertical segment, and integer $\alpha \in [1,n]$.
The expected size of $|R(q) \cap \mathcal{T}_{\leq \alpha}|$ is at most $1+k \tfrac \alpha  n$.
\end{lemma}

\begin{proof}
For any $\Delta \in \mathcal{T}_{\leq n}$, the probability that $\pi(\Delta) \in [1,\alpha]$ is exactly $\alpha/n$ since the segment defining the bottom boundary of $\Delta$ is uniform over $\{1, \ldots, n\}$.
Hence, summing over $K_q$ yields that the expected number of $\alpha$-active points is $\leq 1+k \tfrac \alpha n$.
\end{proof}

Note that $q_0$ is active for each integer $\alpha$. 
\begin{lemma}\label{alpha-active-points}
    The expected number of nodes in $V_q$ with priority in $(\alpha,2\alpha]$ is  ${\cal O}(1+ k\alpha / n)$.
\end{lemma}

\begin{proof}
Let $p\in K_q \setminus \{q_0\}$ be a $2\alpha$-active point. 
For each $i\in\{\alpha+1,...,2\alpha\}$, denote by $A_i$ the leaf node whose region contained $p$ after $i$ segments had been inserted. 
Following above's backward analysis argument, we have that the expected number of nodes with priorities in $(\alpha, 2\alpha]$ that contain $p$ is 
    $\sum_{i=\alpha+1}^{2\alpha}\O(1)/i = \mathcal{O}(1)\log(2\alpha/\alpha) = \mathcal{O}(1)$.

Let $\Delta(N)$ denote the trapezoidal region of TSD node $N$.
For $p \in K_q\setminus\{q_0\}$, we define the random variables
\begin{align*}
    X_p &:= \begin{cases}
    1,& \text{if $p$ is $2\alpha$-active}\\
    0,              & \text{otherwise}
    \end{cases} \\
    Y_p &:= \Big|\big\{N\in V_q ~:~ p\in \Delta(N)~\wedge~ \pi(N)\in (\alpha, 2\alpha]\big\}\Big| \quad.
\end{align*}

The argument made above shows that $\mathbb{E}[Y_p|X_p=1]=\O(1)$.
Hence, we have 
\begin{align*}
    \mathbb{E}[X_p\cdot Y_p] &= 
        \Pr[X_p=1]     (1\cdot \mathbb{E}[ Y_p|X_p=1 ] ) + 
        \Pr[X_p=0]   (0\cdot \mathbb{E}[ Y_p|X_p=0 ] ) \\
    &= \Pr[X_p=1 ]   (1\cdot \mathbb{E}[ Y_p|X_p=1] ) 
    = \O(\alpha/n) \quad.
\end{align*}
Summing over the expectation bounds of the points in $K_q$, the expected number of nodes in $V_q$ of priority in the range $(\alpha, 2\alpha]$ is $\mathcal{O}(1) + \sum_{p \in K_q \setminus\{q_0\}}\mathbb{E} [X_p\cdot Y_p ] = \mathcal{O}(1 + k\alpha/n)$.
\end{proof}

We are now equipped to prove our main result.

\begin{proof}[Proof of Theorem \ref{query-bound}]
We have 
\begin{align*}
\mathbb{E} \left|V_q \right| 
&= \sum_{i=0}^{\log n}\mathbb{E} 
\Big| \big\{N: N\text{ is a node of } V_q \text{ with priority in } \left(2^{i-1}, 2^i\right] \big\} \Big| \\
&\leq \sum_{i=0}^{\log n}\mathcal{O} \left(1+k\frac{2^{i}}{n} \right)\\
&= {\cal O}(\log n) + \sum_{i=0}^{ \log n } \frac{{\cal O}( k)}{2^{\log n-i}} 
 ~\leq~ {\cal O}(\log n) + {\cal O}(k) \sum_{i \geq 0} 1/2^i = {\cal O}(k + \log n)~,
\end{align*}
where the first inequality is due to applying Lemma~\ref{alpha-active-points} for each of the ${\cal O}(\log n)$ intervals in the partition of priority values.
\end{proof}

\subsection{Windows queries on connected planar graphs} \label{sec:connected-planar}

For the important case that the segments $S$ are the edges of a connected planar graph, 
we use one TSD over $S$ and a second TSD over segments $S'$, that are rotated by $90$ degrees, for the vertical and horizontal segment-queries respectively.
Note both TSDs have $O(n)$ expected size and their construction takes $\O(n \log^*\!n)$ expected time (see~\cite[Thm.~3]{Seidel91}).
Thus, we can report all $k$ edges that intersect an axis aligned query window using a simple Breadth-First Search~(BFS) on the edges of the planar graph.
That is, starting from an edge that intersect the boundary, the BFS starts at the vertex inside the query window and visits only adjacent edges that are inside the query window in $\mathcal{O}(k)$ time.
We summarize.

\begin{corollary}\label{cor:joachim}
Reporting the $k$ edges from a connected planar graph of size $n$ that intersect an axis-aligned query window takes expected $\mathcal{O}(k+\log n)$ time.
The expected pre-processing time is $\O(n \log^*\!n)$.
\end{corollary}

\subsubsection{Acknowledgements}
The authors want to thank Joachim Gudmundsson for suggesting the BFS query of Corollary~\ref{cor:joachim} and an anonymous reviewer for constructive feedback.
This work was supported under the Australian Research Council Discovery Projects funding scheme (project number DP180102870).

\bibliographystyle{splncs04}
\bibliography{references.bib}

\end{document}